\documentclass[12pt,a4paper]{article}

\usepackage{amsmath,amssymb,amsthm}
\usepackage{bm}
\usepackage{mathtools}

\usepackage[margin=2.5cm]{geometry}
\usepackage{setspace}
\usepackage{parskip}

\usepackage{booktabs}
\usepackage{array}

\usepackage{tikz}
\usepackage{pgfplots}
\pgfplotsset{compat=1.18}
\usetikzlibrary{arrows.meta, patterns, calc, decorations.markings}
\usepgfplotslibrary{fillbetween}

\usepackage[colorlinks=true,linkcolor=blue,citecolor=blue,urlcolor=blue]{hyperref}
\usepackage[capitalise,noabbrev]{cleveref}

\usepackage[authoryear,round]{natbib}
\newtheorem{proposition}{Proposition}

\theoremstyle{definition}
\newtheorem{definition}{Definition}
\theoremstyle{remark}
\newtheorem*{remark}{Remark}

\newcommand{\dd}{\mathrm{d}}

\newcommand{\tw}{\tau_w}

\newcommand{\Gini}{\mathrm{Gini}}

\begin{document}

\title{From Gravity to Confinement:\\ Wealth Redistribution as Optimal
  Drift Design\\ in the Fokker--Planck Framework}

\author{Anders G.\ Fr{\o}seth\thanks{Independent Researcher.
  Email: \href{mailto:indrefjorden@pm.me}{indrefjorden@pm.me}.}}

\date{3 July 2026}

\maketitle

\begin{abstract}
A proportional wealth tax acts as a uniform gravitational field on the
wealth distribution: it shifts the drift of the Fokker--Planck equation
without altering the diffusion, preserving the Gini coefficient at all
finite times.  The same drift-shift symmetry that makes the tax
non-distortionary \citep{Froeseth2026N, Froeseth2026SP} also makes it
non-redistributive through the market channel.  Redistribution requires
breaking this symmetry.  A progressive tax (confining potential) replaces
the Pareto steady state with a thinner-tailed distribution whose Gini is
a closed-form function of the progressivity parameter; source-sink terms
(tax-funded transfers) reshape the density directly.

We formulate optimal redistribution as a control problem for the
Fokker--Planck equation, penalising intervention costs including
migration, evasion, and portfolio distortion.  In general equilibrium
the tax design feeds back through aggregate capital and the production
function, yielding a self-consistent McKean--Vlasov equation with
diminishing returns to progressivity.  The spectral gap of the
Fokker--Planck operator determines convergence speed: progressive taxes
redistribute within policy-relevant timescales, whereas proportional
taxes rely on slow demographic turnover.


\end{abstract}

\section{Introduction}\label{sec:intro}

This paper studies the redistributive consequences of wealth tax design
within the Fokker--Planck framework for the wealth distribution.

The starting point is the neutrality result established in the companion
papers.  \citet{Froeseth2026N} showed that a proportional wealth tax
levied at market value is neutral with respect to portfolio choice:
it reduces expected wealth and risk by the same proportion, leaving
the risk--reward profile of every portfolio unchanged.
\citet{Froeseth2026SP} reformulated this in Fokker--Planck language:
the proportional tax enters as a uniform reduction of the drift
coefficient, preserving the diffusion structure entirely.  This
\emph{drift-shift symmetry} is the mathematical content of neutrality.

Neutrality entails a corollary for redistribution.  The Gini
coefficient is scale-invariant: a uniform drift shift produces a
proportional scaling of all wealth levels at finite times, leaving
relative inequality unchanged.  The Pareto exponent of the
steady-state distribution does eventually respond, but the transition
takes decades to centuries \citep{GabaixEtAl2016}, far beyond typical
policy horizons.  A tax that does not distort therefore does not
redistribute \emph{through the market channel}.

A proportional wealth tax does, however, redistribute through the
\emph{fiscal channel}: tax revenue funds government expenditure that
benefits lower-wealth citizens disproportionately.  In Fokker--Planck
terms, the fiscal channel corresponds to a source-sink modification
(Class~3 in \Cref{sec:taxonomy}), operating outside the drift-shift
symmetry.

The analysis develops the \emph{homogeneous-returns} case: all
investors face a common drift~$v$ and diffusion~$D$.  This is the
setting of the companion papers \citep{Froeseth2026N, Froeseth2026SP}.
Empirical evidence shows that returns are in fact heterogeneous
\citep{FagerengEtAl2020}; \citet{Froeseth2026H} analyses this case and
shows that the redistribution paradox is specific to the homogeneous
limit.  The implications for the results below are discussed in
\Cref{remark:homogeneous_scope} and \S\ref{sec:discussion}.

The Fokker--Planck framework provides five capabilities that motivate
its use throughout this paper.
First, \emph{exact symmetry results}: the drift-shift symmetry of the
Fokker--Planck propagator yields both neutrality and non-redistribution
as theorems, and identifies precisely which symmetry must be broken to
achieve redistribution.
Second, an \emph{intrinsic taxonomy} of tax designs: every modification
of the wealth process maps to a specific term in the equation (drift,
diffusion, source-sink, or boundary condition), each producing a
qualitatively different steady-state distribution.
Third, \emph{convergence rates} from the spectral theory of the
Fokker--Planck operator: the eigenvalues yield explicit timescales for
how fast inequality responds to a policy change---the calculation
underlying the slow transition dynamics of \citet{GabaixEtAl2016}.
Fourth, \emph{unity of theory and estimation}: the Fokker--Planck
equation is the forward equation of the stochastic differential
equation generating the wealth data, so that neural stochastic
differential equation methods~\citep{SongEtAl2021} estimate exactly
the drift and diffusion coefficients that appear in the optimal policy
problem.
Fifth, the general equilibrium extension maps to a
\emph{McKean--Vlasov} equation with established existence and uniqueness
theory, rather than requiring the numerical self-consistency methods
(e.g.\ \citealp{KrusellSmith1998}) standard in macroeconomics.

This paper models both the market and fiscal channels and asks: for a
given level of wealth tax revenue, what is the optimal tax
\emph{design}?  Different tax structures correspond to different
modifications of the Fokker--Planck equation, with qualitatively
different effects on the steady-state wealth distribution:

\begin{itemize}
  \item A \textbf{proportional tax} (uniform drift shift) preserves the
    Pareto tail.  It is the analogue of a uniform gravitational field:
    every particle (investor) is pulled downward at the same rate,
    regardless of position.

  \item A \textbf{progressive tax} (state-dependent drift, increasing with
    wealth) acts as a \emph{confining potential}---a harmonic trap in
    physical language.  It replaces the Pareto steady state with a
    thinner-tailed distribution, potentially Gaussian in the
    Ornstein--Uhlenbeck limit.  This is the Fokker--Planck realisation of
    \citeauthor{Piketty2014}'s (\citeyear{Piketty2014}) proposal.

  \item \textbf{Explicit transfers} (source-sink terms in the
    Fokker--Planck equation) directly reshape the density by removing
    probability current at high wealth and injecting it at low wealth.
    This is the most aggressive intervention, corresponding to a
    universal basic income funded by a wealth tax.

  \item An \textbf{absorbing boundary} (wealth cap) truncates the Pareto
    tail entirely above a threshold, equivalent to a 100\% marginal tax
    rate.
\end{itemize}

For each class, we derive the resulting steady-state distribution and
its Gini coefficient as explicit functions of the policy parameters.
We then formulate the \emph{optimal drift design problem}: find the
modification of the Fokker--Planck drift that minimises the distance
between the resulting wealth distribution and a target, subject to a
regularisation penalty reflecting the economic and political costs of
intervention.  This is a well-posed optimal control problem for a
partial differential equation, drawing on techniques from both physics
and applied mathematics.

The paper connects three bodies of work that have developed largely
independently:

\begin{enumerate}
  \item \citeauthor{Piketty2014}'s (\citeyear{Piketty2014}) thesis that
    the gap between the return on capital $r$ and economic growth $g$
    drives wealth concentration.  In our framework, $r - g$ determines
    the drift of the Fokker--Planck equation: when $r > g$, the drift is
    positive and the Pareto tail fattens.

  \item \citeauthor{GabaixEtAl2016}'s (\citeyear{GabaixEtAl2016})
    analysis of random multiplicative growth as the mechanism generating
    Pareto tails, and their finding that transition dynamics are
    inherently slow.  This provides the baseline Fokker--Planck equation
    and explains why proportional taxes are ineffective on policy-relevant
    timescales.

  \item The econophysics literature on wealth distribution
    \citep{BouchaudMezard2000, DragulescuYakovenko2000,
    ChatterjeeCM2004}, which applies statistical mechanics to economic
    dynamics but has not fully connected to the optimal taxation
    literature.
\end{enumerate}

Empirical evidence from Norwegian wealth register data
\citep{FagerengEtAl2020} shows that returns to wealth are increasing in
wealth level, with an 18~percentage point gap between the 10th and 90th
percentiles of the net worth distribution.  In Fokker--Planck language, the drift $v(x)$ is an
increasing function of log-wealth~$x$---the system has a built-in
anti-redistributive force that a uniform drift shift cannot overcome.
\citet{FagerengGuisoRing2025} further show that portfolio responses to
equity premium changes (induced by Norwegian wealth tax reforms) are
strong but slow, with adjustment taking 5--6 years---direct evidence
of the relaxation timescales predicted by the Fokker--Planck framework.

A growing empirical literature exploits cross-country variation in
wealth taxation to quantify behavioural responses.
\citet{BrulhartEtAl2022} use Swiss cantonal variation to estimate a
taxable wealth elasticity of approximately 0.43, decomposed into
migration, house price capitalisation, and real behavioural response.
\citet{JakobsenEtAl2020} exploit the 1989 Danish tax reform (a
1.2~percentage point cut in the marginal wealth tax rate) to show that
wealth of moderately wealthy households increased by approximately
20\% over eight years, with even larger effects for the very
wealthy.\footnote{The 8-year treatment-on-the-treated effect on log
wealth is 0.186 for couples in the exempted range and 0.312 for
households unbound by the tax ceiling \citep[Table~II]{JakobsenEtAl2020}.}  \citet{MartiMartinezScheuer2023} find that progressive wealth
tax cuts in Switzerland explain 20--25\% of increases in top wealth
concentration.  These empirical magnitudes are precisely the parameters
needed to calibrate the regularisation term in our optimal control
framework (\Cref{sec:optimal}).

The remainder of the paper is organised as follows.
\Cref{sec:baseline} establishes the baseline: random multiplicative
growth with demographic turnover generates a Pareto-tailed steady state
whose exponent is determined by the drift, diffusion, and turnover rate.
\Cref{sec:neutrality_redistribution} proves that a uniform drift shift
(proportional tax) preserves the Gini coefficient at all finite times
through the market channel, and affects the steady state only on
timescales of decades to centuries.
\Cref{sec:confinement} introduces the confining potential (progressive
tax) and derives the resulting steady state and Gini.
\Cref{sec:taxonomy} presents the full taxonomy of Fokker--Planck
modifications for redistribution.
\Cref{sec:optimal} formulates the optimal drift design problem, first
in partial equilibrium and then in general equilibrium where the drift
responds endogenously to the wealth distribution through the production
function.
\Cref{sec:empirical} discusses the empirical strategy for estimating
the current drift-diffusion structure from wealth register data,
including identification of the general equilibrium feedback.
\Cref{sec:discussion} concludes with a discussion of normative
choices, limitations, and connections to the broader wealth tax
programme.

\section{Random growth and the Pareto steady state}\label{sec:baseline}

\subsection{The baseline Fokker--Planck equation}

Consider a population of individuals whose wealth $W_i(t)$ each evolves
according to geometric Brownian motion:
\begin{equation}\label{eq:gbm}
  \frac{\dd W}{W} = \mu \, \dd t + \sigma \, \dd B_t \,,
\end{equation}
where $\mu$ is the expected return on capital, $\sigma > 0$ is the
volatility, and $B_t$ is a standard Brownian motion.  This is the
workhorse model of finance \citep{Merton1969, BlackScholes1973} and
the starting point of the random growth literature on wealth
distribution \citep{Gabaix1999, GabaixEtAl2016, BenhabibBisinZhu2011}.

Applying It\^{o}'s lemma to $x = \ln W$ gives the log-wealth dynamics:
\begin{equation}\label{eq:logwealth}
  \dd x = v \, \dd t + \sigma \, \dd B_t \,,
  \qquad v \equiv \mu - \tfrac{\sigma^2}{2} \,.
\end{equation}
This is a Langevin equation with constant drift $v$ and constant
diffusion coefficient $D = \sigma^2/2$.  The probability density
$\pi(x,t)$ of log-wealth across the population satisfies the
Fokker--Planck equation:
\begin{equation}\label{eq:fp}
  \frac{\partial \pi}{\partial t}
  = -v \frac{\partial \pi}{\partial x}
  + D \frac{\partial^2 \pi}{\partial x^2} \,.
\end{equation}

Without additional mechanisms, the solution to \eqref{eq:fp} is a
Gaussian that drifts to the right (when $v > 0$) and spreads diffusively:
\begin{equation}\label{eq:propagator}
  \pi(x,t) = \frac{1}{\sqrt{4\pi D t}}
  \exp\!\left(-\frac{(x - x_0 - vt)^2}{4Dt}\right) .
\end{equation}
This propagator captures individual-level dynamics well, but it does not
produce the heavy-tailed wealth distributions observed empirically.  A
stationary distribution requires an additional mechanism.

\subsection{Demographic turnover and the Pareto tail}\label{sec:pareto}

Following \citet{Gabaix1999} and \citet{GabaixEtAl2016}, we introduce
demographic turnover: at rate $\delta > 0$, each individual is replaced
by a new entrant drawn from a distribution concentrated near the
population mean.  This adds a source-sink term to the Fokker--Planck
equation:
\begin{equation}\label{eq:fp_turnover}
  \frac{\partial \pi}{\partial t}
  = -v \frac{\partial \pi}{\partial x}
  + D \frac{\partial^2 \pi}{\partial x^2}
  - \delta \pi + \delta \, \phi(x) \,,
\end{equation}
where $\phi(x)$ is the entrant distribution (normalised).  The term
$-\delta\pi$ removes individuals at rate $\delta$; the term
$+\delta\phi$ inserts replacements.

Setting $\partial\pi/\partial t = 0$ and looking for solutions of the
form $\pi_{\mathrm{ss}}(x) \propto e^{-\alpha x}$ in the right tail
(where $\phi(x) \approx 0$) yields the characteristic equation:
\begin{equation}\label{eq:pareto_eq}
  D\alpha^2 - v\alpha - \delta = 0 \,.
\end{equation}
The positive root gives the Pareto exponent:
\begin{equation}\label{eq:pareto_exponent}
  \alpha = \frac{v + \sqrt{v^2 + 4D\delta}}{2D} \,.
\end{equation}
Since $\pi_{\mathrm{ss}}(x) \propto e^{-\alpha x}$ in log-wealth
corresponds to $p(W) \propto W^{-(1+\alpha)}$ in wealth, this is a
Pareto distribution with tail exponent $\alpha$.  Empirically,
$\alpha \approx 1.5$ for wealth distributions in developed countries
\citep{Gabaix2009}.

\subsection{Gini coefficient and the Pareto exponent}

For a pure Pareto distribution with exponent $\alpha > 1$, the Gini
coefficient is:
\begin{equation}\label{eq:gini_pareto}
  \Gini = \frac{1}{2\alpha - 1} \,.
\end{equation}
For $\alpha = 1.5$, this gives $\Gini = 0.5$.  More generally, the
relationship is monotonically decreasing: a larger $\alpha$ (thinner
tail) corresponds to lower inequality.

\begin{remark}[Piketty's $r > g$ in Fokker--Planck language]
\citeauthor{Piketty2014}'s central thesis is that when the return on
capital $r$ exceeds economic growth $g$, wealth concentrates.  In the
Fokker--Planck framework, this translates directly: the drift
$v = \mu - \sigma^2/2$ determines how fast wealth accumulates relative
to the mean, and the Pareto exponent $\alpha$ depends on $v/D$ and
$\delta/D$.  When $v$ is large (high $r - g$), $\alpha$ is large but
the denominator $2D$ is also large if volatility is high.  The precise
balance between return, volatility, and demographic turnover determines
the steady-state level of inequality.  Piketty's proposal of a
progressive capital tax is, in Fokker--Planck language, a proposal to
replace the uniform drift with a \emph{confining potential}---the
subject of \Cref{sec:confinement}.
\end{remark}

\section[Why proportional taxes cannot redistribute through the
  market channel]{Why proportional taxes cannot redistribute\\
  through the market channel}
\label{sec:neutrality_redistribution}

\subsection{The drift-shift symmetry}

A proportional wealth tax at rate $\tw$ reduces the instantaneous return
from $\mu$ to $\mu - \tw$, giving the taxed log-wealth dynamics:
\begin{equation}\label{eq:taxed_drift}
  \dd x = v_\tau \, \dd t + \sigma \, \dd B_t \,,
  \qquad v_\tau \equiv v - \tw = \mu - \tw - \tfrac{\sigma^2}{2} \,.
\end{equation}
The diffusion coefficient $D = \sigma^2/2$ is unchanged.  The taxed
Fokker--Planck equation is identical to \eqref{eq:fp_turnover} with
$v$ replaced by $v_\tau$: a pure drift shift.  This is the
\emph{drift-shift symmetry} of \citet{Froeseth2026SP}.

\subsection{Scale invariance of the Gini coefficient}

The proportional tax scales all wealth by the same factor:
$W_\tau(t) = (1 - \tw)^t \cdot W_0(t)$ relative to the untaxed path,
where $W_0(t)$ is the wealth that would have obtained without the tax.
At any finite time, the taxed wealth distribution is a rescaled version
of the untaxed distribution.

\begin{proposition}[Gini preservation under uniform drift shift]
\label{prop:gini_invariance}
Let $\pi(x,t)$ solve the Fokker--Planck equation \eqref{eq:fp} with
drift $v$ and initial condition $\pi_0(x)$, and let
$\pi_\tau(x,t)$ solve the same equation with drift $v_\tau = v - \tw$.
Then:
\begin{equation}
  \pi_\tau(x,t) = \pi(x + \tw t, \, t) \,.
\end{equation}
In wealth space, this is a multiplicative rescaling:
$p_\tau(W,t) = p(W \cdot e^{\tw t}, t) \cdot e^{\tw t}$.  Since the
Gini coefficient is invariant under multiplicative rescaling of all
wealth levels:
\begin{equation}\label{eq:gini_preserved}
  \Gini[\pi_\tau(\cdot, t)] = \Gini[\pi(\cdot, t)]
  \qquad \forall\, t \geq 0.
\end{equation}
\end{proposition}

\begin{proof}
Define $\tilde{\pi}(x,t) = \pi(x + \tw t, t)$.  By direct
computation:
\begin{align}
  \frac{\partial \tilde{\pi}}{\partial t}
  &= \tw \frac{\partial \pi}{\partial x}\bigg|_{x+\tw t}
     + \frac{\partial \pi}{\partial t}\bigg|_{x+\tw t} \notag \\
  &= \tw \frac{\partial \pi}{\partial x}\bigg|_{x+\tw t}
     - v \frac{\partial \pi}{\partial x}\bigg|_{x+\tw t}
     + D \frac{\partial^2 \pi}{\partial x^2}\bigg|_{x+\tw t} \notag \\
  &= -(v - \tw) \frac{\partial \tilde{\pi}}{\partial x}
     + D \frac{\partial^2 \tilde{\pi}}{\partial x^2} \,.
\end{align}
So $\tilde{\pi}$ solves the Fokker--Planck equation with drift
$v - \tw = v_\tau$.  Since the Gini coefficient depends only on the
shape of the Lorenz curve, which is invariant under translations in
log-wealth (equivalently, multiplicative rescaling of wealth),
\eqref{eq:gini_preserved} follows.
\end{proof}

\subsection{The steady-state response is real but slow}

While the Gini is preserved at all finite times, the \emph{steady-state}
Pareto exponent does change:
\begin{equation}\label{eq:alpha_taxed}
  \alpha_\tau
  = \frac{v_\tau + \sqrt{v_\tau^2 + 4D\delta}}{2D}
  = \frac{(v - \tw) + \sqrt{(v - \tw)^2 + 4D\delta}}{2D} \,.
\end{equation}
If $\tw > 0$ reduces $v_\tau$ relative to $v$, the steady-state
$\alpha_\tau$ decreases (thicker tail, higher steady-state Gini) when
the drift remains positive, and increases when the drift turns negative.
But as shown by \citet{GabaixEtAl2016}, the convergence to this new
steady state is governed by the spectral gap of the Fokker--Planck
operator, which is of order $\delta$ (the demographic turnover rate).
For realistic parameters ($\delta \approx 1/30$ per year), the
half-life of the transition is:
\begin{equation}\label{eq:halflife}
  t_{1/2} = \frac{\ln 2}{\Lambda} \,,
  \qquad \Lambda \approx \delta + \frac{v_\tau^2}{4D}
  \qquad (\text{for } v_\tau < 0) \,,
\end{equation}
where $\Lambda$ denotes the spectral gap of the Fokker--Planck
operator (not to be confused with the regularisation parameter
$\lambda$ introduced in \Cref{sec:optimal}).
For a $2\%$ tax with $\sigma = 30\%$ and $\mu = 8\%$, this gives
$t_{1/2} \approx 21$ years.  The distribution responds, but on
timescales far beyond typical electoral cycles.

\begin{remark}[The redistribution paradox]
\Cref{prop:gini_invariance} gives the market-channel result:
the drift-shift symmetry that makes the proportional tax
non-distortionary also makes it non-redistributive.  As
established in the introduction, the fiscal channel
(government spending funded by tax revenue) operates as a
source-sink term outside this symmetry and \emph{is}
redistributive---see Class~3 in \Cref{sec:taxonomy}.  The
paradox, then, is precise: the symmetry that makes the tax
optimal for portfolio neutrality is exactly the symmetry that
makes the market channel powerless.  The remainder of this
paper asks how to break this symmetry---through progressive
design of the market channel---while preserving the fiscal
channel that a proportional base already provides.

Breaking the drift-shift symmetry through progressivity
(\Cref{sec:confinement}) adds a second, market-based redistribution
mechanism on top of the fiscal channel.
\end{remark}

\begin{remark}[Scope: homogeneous returns]
\label{remark:homogeneous_scope}
As stated in the introduction, \Cref{prop:gini_invariance} and the
redistribution paradox rely on the homogeneous-returns assumption.
When persistent ability $z$ generates an investor-specific drift
$v(z)$, the uniform shift $v(z) \to v(z) - \tw$ imposes heterogeneous
\emph{relative} burdens: low-ability investors may face negative net
drift while high-ability investors remain in positive territory.
Over time, wealth reallocates from low- to high-ability types---a
market-channel redistribution mechanism absent in the homogeneous case.

This ``use-it-or-lose-it'' effect \citep{Guvenen2023, Froeseth2026H}
means the redistribution paradox is specific to the homogeneous limit.
Under heterogeneous returns, a proportional wealth tax \emph{is}
redistributive through the market channel---not by compressing the
distribution (which requires the confining potential of
\Cref{sec:confinement}), but by accelerating the reallocation of
capital toward higher-ability investors.  The optimal drift design
framework of \Cref{sec:optimal} remains the appropriate tool for
\emph{active} redistribution beyond this passive reallocation effect.
\end{remark}

\section{From gravity to confinement: the progressive tax}
\label{sec:confinement}

\subsection{The confining potential}

The proportional tax enters the Fokker--Planck equation as a uniform
drift shift $v \to v - \tw$---a constant external force, analogous to
a uniform gravitational field.  Gravity shifts the distribution without
confining it.

A \emph{progressive} wealth tax, where the effective rate increases with
wealth, introduces a \emph{state-dependent} drift:
\begin{equation}\label{eq:progressive}
  v(x) = v_0 - \kappa(x - \bar{x}) \,,
\end{equation}
where $v_0 = \mu - \sigma^2/2 - \tw^{(0)}$ is the drift at the
reference log-wealth $\bar{x}$, $\tw^{(0)}$ is the base tax rate,
and $\kappa > 0$ is the \emph{progressivity parameter}.  The term
$-\kappa(x - \bar{x})$ is a \emph{restoring force}: it pushes
high-wealth individuals ($x > \bar{x}$) leftward more strongly than
low-wealth individuals.

In physics, this is a \emph{harmonic potential}: the force increases
linearly with displacement from the reference point.  The resulting
dynamics are those of an Ornstein--Uhlenbeck (OU) process:
\begin{equation}\label{eq:ou}
  \dd x = [v_0 - \kappa(x - \bar{x})] \, \dd t + \sigma \, \dd B_t \,.
\end{equation}

\subsection{The Ornstein--Uhlenbeck steady state}

The Fokker--Planck equation for \eqref{eq:ou} is:
\begin{equation}\label{eq:fp_ou}
  \frac{\partial \pi}{\partial t}
  = -\frac{\partial}{\partial x}
    \bigl\{[v_0 - \kappa(x - \bar{x})]\,\pi\bigr\}
  + D \frac{\partial^2 \pi}{\partial x^2} \,.
\end{equation}
Setting $\partial\pi/\partial t = 0$ and solving gives the
steady-state distribution:
\begin{equation}\label{eq:ou_ss}
  \pi_{\mathrm{ss}}(x)
  = \frac{1}{\sqrt{2\pi \cdot D/\kappa}}
  \exp\!\left(
    -\frac{(x - \bar{x} - v_0/\kappa)^2}{2D/\kappa}
  \right) .
\end{equation}
This is a \emph{Gaussian} in log-wealth, centred at
$x^* = \bar{x} + v_0/\kappa$ with variance $\sigma_{\mathrm{ss}}^2
= D/\kappa$.

The transformation is dramatic: the Pareto tail (power-law decay)
is replaced by a Gaussian tail (exponential-quadratic decay).  In
wealth space, the steady state becomes \emph{log-normal}---a
distribution with much thinner tails than the Pareto.

\subsection{Gini coefficient under confinement}

For a log-normal distribution with log-variance $\sigma_{\mathrm{ss}}^2
= D/\kappa$, the Gini coefficient is:
\begin{equation}\label{eq:gini_lognormal}
  \Gini_{\mathrm{OU}}
  = 2\Phi\!\left(\frac{\sigma_{\mathrm{ss}}}{\sqrt{2}}\right) - 1
  = 2\Phi\!\left(\sqrt{\frac{D}{2\kappa}}\right) - 1 \,,
\end{equation}
where $\Phi$ is the standard normal CDF.

This is a monotonically decreasing function of the progressivity
parameter $\kappa$: more progressivity confines the distribution more
tightly, reducing the Gini.

\begin{proposition}[Gini as a function of progressivity]
\label{prop:gini_kappa}
For the Ornstein--Uhlenbeck wealth dynamics \eqref{eq:ou} with
diffusion coefficient $D = \sigma^2/2$ and progressivity $\kappa > 0$:
\begin{enumerate}
  \item[(i)] $\Gini_{\mathrm{OU}} \to 0$ as $\kappa \to \infty$
    (complete equality in the limit of infinite progressivity).
  \item[(ii)] $\Gini_{\mathrm{OU}} \to 1$ as $\kappa \to 0^+$
    (the distribution spreads without bound, approaching maximal
    inequality).
  \item[(iii)] For a target Gini $G^* \in (0,1)$, the required
    progressivity is:
    \begin{equation}\label{eq:kappa_target}
      \kappa^* = \frac{D}{2\bigl[\Phi^{-1}\!\bigl(\tfrac{1+G^*}{2}
        \bigr)\bigr]^2} \,.
    \end{equation}
\end{enumerate}
\end{proposition}

\begin{proof}
Part (i): As $\kappa \to \infty$, $\sigma_{\mathrm{ss}}^2 = D/\kappa
\to 0$, so the log-normal concentrates at a point and
$\Gini \to 0$.  Part (ii): As $\kappa \to 0^+$, $\sigma_{\mathrm{ss}}
\to \infty$ and $\Gini \to 1$.  Part (iii): Solve
$2\Phi(\sqrt{D/(2\kappa)}) - 1 = G^*$ for $\kappa$.
\end{proof}

\begin{remark}[Numerical illustration]
For typical parameters ($\sigma = 0.30$, so $D = 0.045$), achieving
the current empirical Gini of $\approx 0.80$ requires
$\kappa \approx 0.014$; reducing the Gini to $0.50$ (the level of a
pure exponential distribution, \`{a} la \citealt{DragulescuYakovenko2000})
requires $\kappa \approx 0.050$; and reaching the Scandinavian income
Gini of $\approx 0.27$ would require $\kappa \approx 0.19$.  These
values correspond to marginal tax rate increases of approximately
$1.4\%$, $5.0\%$, and $19\%$ per unit of log-wealth, respectively.
\end{remark}

\subsection{Convergence speed: the advantage of confinement}

The OU process has a well-defined spectral gap equal to $\kappa$: all
deviations from the steady state decay at rate $e^{-\kappa t}$.  The
half-life of convergence is:
\begin{equation}\label{eq:ou_halflife}
  t_{1/2}^{\mathrm{OU}} = \frac{\ln 2}{\kappa} \,.
\end{equation}
For $\kappa = 0.05$, this is $t_{1/2} \approx 14$ years---comparable
to the demographic turnover timescale, but now driven by the \emph{tax
itself} rather than by births and deaths.  The confining potential
actively compresses the distribution, rather than passively waiting for
demographic replacement.

This is a qualitative difference from the proportional tax.
\Cref{fig:convergence} compares the convergence trajectories: the
proportional tax relies on demographic turnover (slow), while the
progressive tax creates its own convergence mechanism (fast).

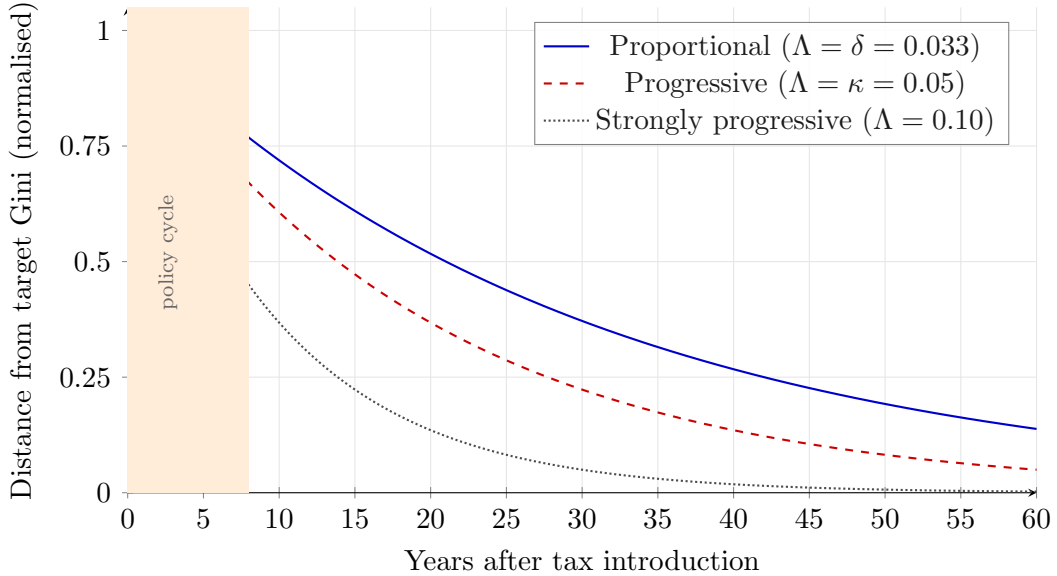
\begin{figure}[t]
\centering
\begin{tikzpicture}
\begin{axis}[
  width=0.85\textwidth,
  height=0.50\textwidth,
  xlabel={Years after tax introduction},
  ylabel={Distance from target Gini (normalised)},
  domain=0:60,
  samples=200,
  ymin=0,
  ymax=1.05,
  xmin=0,
  xmax=60,
  legend style={
    at={(0.97,0.97)},
    anchor=north east,
    font=\small,
    draw=black!50,
    fill=white,
    fill opacity=0.9,
  },
  every axis plot/.append style={thick},
  axis lines=left,
  clip=true,
  tick label style={font=\small},
  label style={font=\small},
  ytick={0, 0.25, 0.5, 0.75, 1.0},
  grid=major,
  grid style={gray!20},
]

\addplot[blue!80!black, solid]
  {exp(-0.033*x)};
\addlegendentry{Proportional ($\Lambda = \delta = 0.033$)}

\addplot[red!80!black, dashed]
  {exp(-0.05*x)};
\addlegendentry{Progressive ($\Lambda = \kappa = 0.05$)}

\addplot[black!70, densely dotted]
  {exp(-0.10*x)};
\addlegendentry{Strongly progressive ($\Lambda = 0.10$)}

\fill[orange!15] (axis cs:0,0) rectangle (axis cs:8,1.05);
\node[font=\scriptsize, black!60, rotate=90, anchor=south]
  at (axis cs:4,0.52) {policy cycle};

\end{axis}
\end{tikzpicture}
\caption{Convergence to target wealth distribution under different tax
  structures.  The proportional tax (blue, solid) converges at the
  demographic turnover rate $\delta$; the progressive tax (red, dashed)
  converges at the progressivity rate $\kappa$; a more strongly
  progressive tax (black, dotted) converges faster still.  The shaded
  region marks a typical electoral cycle (4--8 years).  The progressive
  tax achieves meaningful redistribution within policy-relevant
  timescales; the proportional tax does not.}
\label{fig:convergence}
\end{figure}

\section{Taxonomy of Fokker--Planck modifications for redistribution}
\label{sec:taxonomy}

We now present the full classification of policy tools by their
Fokker--Planck modification, steady-state distribution, and effect on
inequality.

\subsection{Class 1: Uniform drift shift (proportional tax)}

\textbf{FP modification:} $v \to v - \tw$ (constant shift).

\textbf{Steady state:} Pareto with modified exponent $\alpha_\tau$
given by \eqref{eq:alpha_taxed}.

\textbf{Gini:} Preserved at finite times; slow convergence to
$1/(2\alpha_\tau - 1)$.

\textbf{Physics analogue:} Uniform gravitational field.

\subsection{Class 2: State-dependent drift (progressive tax)}

\textbf{FP modification:} $v \to v_0 - \kappa(x - \bar{x})$
(linear restoring force).

\textbf{Steady state:} Log-normal (Gaussian in log-wealth) with
variance $D/\kappa$.

\textbf{Gini:} $2\Phi(\sqrt{D/(2\kappa)}) - 1$, continuously
adjustable via $\kappa$.

\textbf{Physics analogue:} Harmonic potential / Ornstein--Uhlenbeck
confinement.

\textbf{Convergence rate:} $\kappa$ (actively driven by the tax).

\subsection{Class 3: Source-sink redistribution (transfers)}

\textbf{FP modification:} Add $-\gamma(x)\pi + S(x)$ where $\gamma(x)$
is a wealth-dependent extraction rate and $S(x)$ is a transfer
injection.

\textbf{Steady state:} Depends on the specific forms of $\gamma$ and
$S$.  For a flat UBI funded by progressive extraction, the steady state
acquires a lower bound (minimum wealth) and a thinner upper tail.

\textbf{Gini:} Can be reduced to any target level by choosing
$\gamma$ and $S$ appropriately.

\textbf{Physics analogue:} Particle removal and injection; coupled
reservoirs.

\textbf{Note:} This class captures the \emph{fiscal channel} of
wealth taxation.  Any wealth tax---including a proportional one---generates
revenue that funds government expenditure.  If that expenditure
benefits citizens at lower wealth levels disproportionately (through
public services, welfare transfers, or infrastructure), the net effect
is a source-sink modification of the wealth distribution.  In this
sense, even the neutral proportional tax of Class~1 is redistributive
when combined with its fiscal counterpart: the drift shift preserves
the market-channel Gini, but the tax-and-transfer flow reduces the
\emph{total} Gini.  The revenue from a 2\% wealth tax on
billionaires, as proposed by \citet{SaezZucman2019} and
\citet{Zucman2024}, would fund substantial source terms at low $x$.

\subsection{Class 4: Absorbing boundary (wealth cap)}

\textbf{FP modification:} $\pi(x^*, t) = 0$ for $x > x^*$, with
probability current at $x^*$ recycled to lower wealth.

\textbf{Steady state:} Truncated distribution with hard upper bound.

\textbf{Gini:} Bounded above by $\Gini_{\mathrm{Pareto}}(x^*)$, the
Gini of a Pareto distribution truncated at $x^*$.

\textbf{Physics analogue:} Absorbing boundary with particle recycling.

\subsection{Class 5: State-dependent diffusion (volatility policy)}

\textbf{FP modification:} $D \to D(x)$ with $D(x)$ decreasing at
high wealth.

\textbf{Steady state:} Thinner tail than Pareto (the tail exponent
increases where $D$ decreases).

\textbf{Gini:} Reduced, but the policy instrument is indirect (reducing
risk for the wealthy, e.g.\ through portfolio insurance or
diversification mandates).

\textbf{Physics analogue:} Spatially varying temperature; cooler
regions have less fluctuation.

\subsection{Summary}

\begin{table}[t]
\centering
\caption{Taxonomy of Fokker--Planck modifications for redistribution.
  Each row represents a class of policy intervention, characterised by
  the modification it introduces into the Fokker--Planck equation and
  its effect on the steady-state wealth distribution.}
\label{tab:taxonomy}
\renewcommand{\arraystretch}{1.4}
\small
\begin{tabular}{@{}llll@{}}
\toprule
\textbf{Policy class} & \textbf{FP modification} &
  \textbf{Steady state} & \textbf{Gini effect} \\
\midrule
Proportional tax & $v \to v - \tw$ &
  Pareto (shifted $\alpha$) & Preserved (finite $t$) \\
Progressive tax & $v(x) = v_0 - \kappa x$ &
  Log-normal & Controlled by $\kappa$ \\
Transfers (UBI) & Source--sink terms &
  Bounded below & Strongly reduced \\
Wealth cap & Absorbing boundary &
  Truncated Pareto & Hard upper bound \\
Volatility policy & $D \to D(x)$ &
  Thinner tail & Indirectly reduced \\
\bottomrule
\end{tabular}
\end{table}

\subsection{Leakage channels: migration and evasion}\label{sec:leakage}

The five classes above describe what the policymaker \emph{intends} to do
to the Fokker--Planck equation.  In practice, two channels cause the
\emph{realised} modification to differ from the intended one.

\paragraph{Migration (permeable boundary).}
If agents can relocate to avoid the tax, the confining potential
$v(x) = v_0 - \kappa(x - \bar{x})$ operates only up to some threshold
$x_m$ beyond which agents exit the jurisdiction.  In Fokker--Planck
language, this replaces the reflecting boundary at infinity with a
\emph{partially absorbing boundary}:
\begin{equation}\label{eq:migration}
  J(x_m, t) = \gamma \, \pi(x_m, t) \,,
\end{equation}
where $J$ is the probability current and $\gamma > 0$ is a migration
rate that increases with the tax burden at $x_m$.
\citet{KlevenEtAl2025} document significant migration responses to
wealth taxation; \citet{Pichet2007} estimates that the French ISF
caused approximately \texteuro200B in capital flight since the tax's creation in 1988.
The consequence is a truncated steady state: the Pareto tail is cut
off not by policy design (Class~4) but by agent exit, and the
effective Gini reduction is smaller than the closed-boundary
prediction.

\paragraph{Tax evasion and avoidance (attenuated drift).}
If agents can hide or reclassify wealth, the intended progressivity
$\kappa$ is not fully realised.  Let $\epsilon(x)$ denote the
evasion rate---the fraction of wealth at position $x$ that escapes
the tax base.  The effective drift modification becomes:
\begin{equation}\label{eq:evasion}
  \delta v_{\mathrm{eff}}(x) = -[1 - \epsilon(x)]\,\kappa\,(x - \bar{x}) \,,
\end{equation}
so the effective spring constant is $\kappa_{\mathrm{eff}}(x) =
[1 - \epsilon(x)]\,\kappa < \kappa$.
\citet{AlstadsaeterJohannesenZucman2019} estimate $\epsilon \approx
0.25$ for the top 0.01\%; \citet{BjornebyEtAl2023} document avoidance
through corporate restructuring that creates a wedge between intended
and realised drift.  Since $\epsilon(x)$ is typically increasing in
$x$ (richer agents have more avoidance opportunities), the effective
confining force weakens precisely where it is most needed.

\paragraph{Debt and leverage (amplified assessment distortion).}
When the wealth tax base is net wealth (assets minus liabilities) and
assets are assessed below market value while debt is deducted at face
value, the tax base diverges from market net wealth.  An investor
with assets of market value~$A$, assessment ratio~$\beta < 1$, and
debt~$D$ has taxable net wealth $W_{\mathrm{tax}} = \beta A - D$
versus market net wealth $W_{\mathrm{market}} = A - D$.  As shown in
\citet{Froeseth2026SP}, the ratio
$W_{\mathrm{tax}} / W_{\mathrm{market}}$ decreases with leverage and
can become negative.  This creates an incentive to borrow against
underassessed assets (notably real estate), which affects the
Fokker--Planck dynamics in two ways: (i)~the effective drift depends
on the leverage ratio, not only on wealth level, coupling the debt
decision to the wealth dynamics; and (ii)~leveraged positions amplify
the effective volatility of net wealth, making the diffusion
coefficient $D$ state-dependent through an endogenous mechanism.
Unlike evasion, which attenuates the intended drift, the leverage
channel can \emph{amplify} the deviation between the observed taxable
distribution and the true market-value distribution.

\begin{remark}[Leakage in the taxonomy]
Migration modifies the \emph{boundary condition} of the
Fokker--Planck equation; evasion modifies the \emph{drift
coefficient}; debt leverage modifies both the \emph{drift} (through
the leverage ratio) and the \emph{diffusion} (through amplified
volatility).  All three reduce the effective redistribution below the
designed level, but through mathematically distinct channels.  The
diffusion model approach of \Cref{sec:empirical} captures all three
implicitly: the learned drift $\hat{v}(x)$ reflects the
\emph{realised} drift including evasion and leverage effects, and
the learned boundary behaviour reflects migration.
\end{remark}

\section{Optimal drift design}\label{sec:optimal}

\subsection{The control problem}

We now formulate the redistribution problem as an optimal control
problem for the Fokker--Planck equation.  The state is the wealth
density $\pi(x,t)$; the control is the drift modification
$\delta v(x)$; the objective is to steer the distribution toward a
target $\pi^*(x)$ with prescribed inequality level.

\begin{definition}[Optimal drift design problem]
\label{def:optimal}
Given a baseline Fokker--Planck equation with drift $v$ and diffusion
$D$, a target distribution $\pi^*$, a time horizon $T$, and a
regularisation parameter $\lambda > 0$, find the drift modification
$\delta v(x)$ that minimises:
\begin{equation}\label{eq:objective}
  \mathcal{J}[\delta v]
  = \frac{1}{2}\int_{-\infty}^{\infty}
    \bigl|\pi(x,T) - \pi^*(x)\bigr|^2 \, \dd x
  + \frac{\lambda}{2}\int_{-\infty}^{\infty}
    |\delta v(x)|^2 \, \dd x \,,
\end{equation}
subject to the controlled Fokker--Planck equation:
\begin{equation}\label{eq:fp_controlled}
  \frac{\partial \pi}{\partial t}
  = -\frac{\partial}{\partial x}
    \bigl\{[v + \delta v(x)]\,\pi\bigr\}
  + D \frac{\partial^2 \pi}{\partial x^2}
  - \delta\pi + \delta\phi \,,
\end{equation}
with initial condition $\pi(x,0) = \pi_0(x)$.
\end{definition}

The first term in \eqref{eq:objective} measures how close the resulting
distribution is to the target; the second term penalises the magnitude
of the intervention.  The regularisation parameter $\lambda$ captures
the trade-off between redistributive effectiveness and the economic and
political costs of intervention.

\subsection{Necessary conditions}

Applying Pontryagin's maximum principle, the optimal drift modification
satisfies:
\begin{equation}\label{eq:optimality}
  \delta v^*(x) = -\frac{1}{\lambda} \pi(x,T) \cdot p(x,T) \,,
\end{equation}
where $p(x,t)$ is the adjoint variable solving the backward equation:
\begin{equation}\label{eq:adjoint}
  -\frac{\partial p}{\partial t}
  = [v + \delta v(x)] \frac{\partial p}{\partial x}
  + D \frac{\partial^2 p}{\partial x^2} - \delta p \,,
\end{equation}
with terminal condition $p(x,T) = \pi(x,T) - \pi^*(x)$.

This is a coupled forward-backward PDE system: the forward equation
propagates the density $\pi$; the backward equation propagates the
adjoint $p$; and the optimal control $\delta v^*$ depends on both.
Standard iterative methods (gradient descent on $\delta v$, or
forward-backward splitting) converge to the solution.

\subsection{The linear-quadratic case}

When the target distribution is log-normal (corresponding to a Gaussian
in log-wealth), and we restrict the control to the linear class
$\delta v(x) = -\kappa(x - \bar{x}) - c$ for constants $\kappa, c$,
the problem reduces to choosing $\kappa$ and $c$ optimally.  The
solution follows from \Cref{prop:gini_kappa}: for a given target Gini
$G^*$, the optimal progressivity is $\kappa^*$ given by
\eqref{eq:kappa_target}, and $c$ adjusts the mean.  This is the
linear-quadratic regulator (LQR) for the Fokker--Planck equation.

\subsection{The regularisation parameter as political feasibility}

The parameter $\lambda$ in \eqref{eq:objective} deserves interpretation.
A large $\lambda$ penalises intervention heavily, yielding a
near-proportional (near-neutral) tax.  A small $\lambda$ permits
aggressive redistribution.  In practice, $\lambda$ encodes:

\begin{itemize}
  \item \textbf{Distortion costs:} The welfare loss from portfolio
    distortion and behavioural responses (\emph{cf.}\ the
    distortion channels of \citealp{Froeseth2026E}).  Empirical
    estimates from Switzerland \citep{BrulhartEtAl2022} and Denmark
    \citep{JakobsenEtAl2020} provide calibration targets.

  \item \textbf{Leakage costs:} The migration and evasion channels
    formalised in \Cref{sec:leakage} reduce the effective
    progressivity below the designed level.  The permeable boundary
    \eqref{eq:migration} truncates the tail through agent exit; the
    attenuated drift \eqref{eq:evasion} weakens the confining force.
    These channels do not simply increase costs---they \emph{modify
    the control problem itself}: the effective control is
    $\kappa_{\mathrm{eff}} = [1 - \epsilon]\,\kappa$, so the
    optimality condition \eqref{eq:optimality} must be evaluated at
    the realised drift, not the legislated one.

  \item \textbf{Administrative costs:} Progressive taxes require
    valuation of all assets at market value, which may be infeasible
    for illiquid assets.

  \item \textbf{Political constraints:} The political economy of
    taxation limits how progressive the tax can be.
\end{itemize}

The optimal control framework makes this trade-off explicit and
quantifiable, rather than leaving it to qualitative argument.  In
particular, the leakage channels imply that the policymaker must
over-design $\kappa$ relative to the closed-system optimum to
compensate for the attenuation by $\epsilon(x)$ and the probability
current lost through migration at $x_m$.

\subsection[General equilibrium: the self-consistent Fokker--Planck
  equation]{General equilibrium: the self-consistent\\
  Fokker--Planck equation}\label{sec:ge}

The analysis so far treats the pre-tax drift $v$ as an exogenous
parameter.  This is the partial equilibrium assumption: the return
on capital $r$ (and hence $v = r - \sigma^2/2$) does not respond to
the tax.  But in general equilibrium, $r$ depends on the aggregate
capital stock $K$, which is itself determined by the wealth
distribution~$\pi$.  The pre-tax return and economic growth $g$ are
jointly determined by technology, preferences, and policy---a point
emphasised by \citet{Jones2015} in the context of Piketty's $r > g$
thesis.

To close the model, we introduce a production function
$Y = F(K, L)$ with diminishing returns to capital.  The marginal
product of capital gives the pre-tax return:
\begin{equation}\label{eq:mpk}
  r = F_K(K, L) \,,
  \qquad K[\pi] = \int_0^\infty W \, p(W) \, \dd W
  = \int_{-\infty}^{\infty} e^x \, \pi(x) \, \dd x \,,
\end{equation}
where $K[\pi]$ is a functional of the wealth density.  The growth
rate $g$ may also depend on $K$ (through capital deepening), so both
$r$ and $g$ respond endogenously to changes in $\pi$.

The drift becomes a functional of the distribution:
\begin{equation}\label{eq:drift_ge}
  v[\pi] = r\bigl(K[\pi]\bigr) - \tfrac{\sigma^2}{2} \,.
\end{equation}
The self-consistent Fokker--Planck equation is then:
\begin{equation}\label{eq:fp_sc}
  \frac{\partial \pi}{\partial t}
  = -\frac{\partial}{\partial x}
    \Bigl\{\bigl[v[\pi] + \delta v(x)\bigr]\,\pi\Bigr\}
  + D \frac{\partial^2 \pi}{\partial x^2}
  - \delta\pi + \delta\phi \,,
\end{equation}
where the drift $v[\pi]$ is itself determined by $\pi$ through
\eqref{eq:mpk}--\eqref{eq:drift_ge}.  This is a \emph{nonlinear}
Fokker--Planck equation: the drift depends on the solution.

\begin{remark}[Mean-field structure]
Equation~\eqref{eq:fp_sc} has the structure of a \emph{McKean--Vlasov}
or mean-field equation, familiar from statistical physics: each
particle (investor) experiences a drift that depends on the aggregate
state of the population.  The self-consistency condition
$v = v[\pi]$ is analogous to the self-consistent field in
Hartree--Fock theory or the mean-field approximation in spin systems.
\citet{BernardBouchaudLeDoussal2026} solve this class of model
explicitly for the case of heterogeneous growth rates with mean-field
redistribution, identifying a phase transition between localised
(condensed) and delocalised (spread) wealth distributions at a critical
redistribution intensity.
\end{remark}

\subsubsection{Partial versus general equilibrium: qualitative
  differences}

The GE feedback introduces a countervailing force.  In partial
equilibrium, a confining potential with progressivity $\kappa$
monotonically reduces the Gini: more $\kappa$ always means less
inequality.  In general equilibrium, increasing $\kappa$ compresses
the wealth distribution, reducing aggregate capital $K$.  By
diminishing returns, $r = F_K(K,L)$ \emph{rises}, increasing the
baseline drift $v$.  This partially offsets the confinement:
\begin{equation}\label{eq:ge_feedback}
  \kappa \uparrow
  \;\Rightarrow\; K \downarrow
  \;\Rightarrow\; r \uparrow
  \;\Rightarrow\; v \uparrow
  \;\Rightarrow\; \text{partial offset of confinement} \,.
\end{equation}
\begin{remark}[Neutrality and economic growth]
The feedback chain \eqref{eq:ge_feedback} rests on the assumption
that wealth taxation reduces productive capital accumulation.  But
\citet{Froeseth2026N} showed that a proportional wealth tax at
market value is \emph{neutral} with respect to portfolio choice:
asset prices, risk-taking, and portfolio allocations are unchanged.
If the progressive tax is constructed as a neutral proportional base
plus a confining surcharge, the base component does not distort
capital allocation.  Only the progressive surcharge $-\kappa(x -
\bar{x})$ introduces distortions.  This means the channel from
taxation to lower $g$ operates only through the progressive
component, not through the full tax rate.  The neutral base reduces
the \emph{level} of wealth but preserves the \emph{composition} of
risk-taking that drives innovation and growth---weakening the
standard objection that wealth taxes discourage entrepreneurship.
In the self-consistent framework, this implies that the GE feedback
\eqref{eq:ge_feedback} may be substantially smaller than a
na\"{\i}ve analysis (which attributes the full tax rate to capital
destruction) would suggest.
\end{remark}

The steady state is determined by the fixed point where the
confining force and the GE feedback balance.  For a Cobb--Douglas
production function $Y = K^a L^{1-a}$ with capital share $a$,
the fixed-point condition is:
\begin{equation}\label{eq:fixed_point}
  v^* = a \bigl(K[\pi^*_{\kappa, v^*}]\bigr)^{a - 1}
    L^{1-a} - \tfrac{\sigma^2}{2} \,,
\end{equation}
where $\pi^*_{\kappa, v^*}$ is the steady state of the FP equation
with drift $v^* - \kappa(x - \bar{x})$.  This is a scalar
fixed-point equation in $v^*$ (since $K$ is a deterministic
functional of the Gaussian steady state), which can be solved by
simple iteration.

\subsubsection{The self-consistent optimal control problem}

The optimal drift design problem (\Cref{def:optimal}) generalises
to the GE setting by replacing the exogenous drift with the
self-consistent drift:
\begin{equation}\label{eq:objective_ge}
  \mathcal{J}_{\mathrm{GE}}[\delta v]
  = \frac{1}{2}\int \bigl|\pi(x,T) - \pi^*(x)\bigr|^2 \, \dd x
  + \frac{\lambda}{2}\int |\delta v(x)|^2 \, \dd x \,,
\end{equation}
subject to \eqref{eq:fp_sc}.  The key difference is that the
forward equation is now nonlinear (through $v[\pi]$), so the adjoint
equation acquires an additional term from the functional derivative
$\delta v / \delta \pi$.  The optimality condition becomes:
\begin{equation}\label{eq:optimality_ge}
  \delta v^*(x)
  = -\frac{1}{\lambda}\biggl[
    \pi(x,T) \cdot p(x,T)
    + \int \frac{\partial v}{\partial K}
      \cdot e^{x'} \pi(x',T) \, p(x',T) \, \dd x'
  \biggr] \,,
\end{equation}
where the second term captures the GE feedback: the optimal policy
accounts for its own effect on the aggregate return.

\subsubsection{Diminishing returns to progressivity}

The GE feedback implies \emph{diminishing returns to progressivity}.
Define the GE Gini reduction per unit of progressivity:
\begin{equation}
  \eta(\kappa) \equiv -\frac{\dd \Gini_{\mathrm{GE}}(\kappa)}
    {\dd \kappa} \,.
\end{equation}
In partial equilibrium, $\eta$ is always positive (more progressivity
always helps).  In general equilibrium, $\eta(\kappa)$ is positive
for small $\kappa$ but decreasing, and may approach zero for large
$\kappa$ if the rise in $r$ fully offsets the confinement.  Whether
$\eta$ actually reaches zero---whether there is a \emph{maximum
achievable redistribution}---depends on the elasticity of the
production function.  For Cobb--Douglas, the fixed point always
exists and $\eta > 0$ for all finite $\kappa$, but the marginal
gain shrinks.

This has a direct policy implication: the GE feedback provides a
natural economic limit on redistribution that is \emph{separate}
from the political constraints encoded in $\lambda$.  Even with
$\lambda = 0$ (no political cost), the production function constrains
how much inequality can be reduced through wealth taxation alone.

\section{Empirical strategy: learning the drift from data}
\label{sec:empirical}

The framework above delivers testable predictions and closed-form
policy formulas, but treats the drift $v$ and diffusion $D$ as known.
In practice, they must be estimated from data.  This section outlines
the estimation strategy; empirical results using wealth register data
will follow in a companion paper.

\subsection{The inverse problem}

Norwegian wealth register data
\citep{FagerengEtAl2020, FagerengGuisoRing2025} provides individual-level
wealth panels spanning two decades, covering the full population.

The estimation problem is: given observations $\{W_i(t_k)\}$ for
individuals $i = 1, \ldots, N$ at times $t_0, t_1, \ldots, t_K$,
recover $v(x)$ and $D(x)$ nonparametrically.

\subsection{Score-based estimation}

A score-based diffusion model learns the score function
$s(x,t) = \partial \log \pi(x,t) / \partial x$ from data.  At
steady state, the score function is related to the FP coefficients by:
\begin{equation}\label{eq:score_fp}
  s_{\mathrm{ss}}(x) = \frac{v(x)}{D(x)} - \frac{D'(x)}{D(x)} \,.
\end{equation}
If $D$ is constant (or independently estimated), the score function
directly recovers the drift:
\begin{equation}\label{eq:drift_from_score}
  v(x) = D \cdot s_{\mathrm{ss}}(x) \,.
\end{equation}

\subsection{Neural SDE approach}

An alternative is to parameterise $v(x)$ and $D(x)$ as neural
networks and train them by maximum likelihood on the panel data,
using differentiable SDE solvers \citep{KidgerDiffrax2021}.  This
directly learns the drift and diffusion functions without going
through the score.

\subsection{Identifying the general equilibrium feedback}

The self-consistent FP equation (\Cref{sec:ge}) predicts that the
drift $v$ responds endogenously to changes in the wealth distribution.
A diffusion model trained on panel data spanning a tax reform captures
this feedback \emph{implicitly}: the learned drift before the reform
reflects the old equilibrium; the drift after reflects the new one,
including whatever adjustment in $r$ occurred through the aggregate
capital channel.

Concretely, let $\hat{v}_{\mathrm{pre}}(x)$ and
$\hat{v}_{\mathrm{post}}(x)$ denote the drift functions estimated
from data before and after a reform that changes the tax rate by
$\Delta\tau$.  In partial equilibrium, the drift change would be a
uniform shift: $\hat{v}_{\mathrm{post}} - \hat{v}_{\mathrm{pre}}
= -\Delta\tau$ for all $x$.  In general equilibrium, the aggregate
capital response adds a positive offset (since $K$ falls and $r$
rises):
\begin{equation}\label{eq:ge_test}
  \hat{v}_{\mathrm{post}}(x) - \hat{v}_{\mathrm{pre}}(x)
  = -\Delta\tau + \Delta r\bigl(K[\pi_{\mathrm{post}}]\bigr) \,,
\end{equation}
where $\Delta r > 0$ is the GE feedback.  This is testable: if the
empirically estimated drift shift is \emph{smaller} than the
mechanical tax change $\Delta\tau$, the GE feedback is present.  The
magnitude of $\Delta r$ identifies the elasticity of the production
function.

\subsection{Natural experiment identification}

The Norwegian wealth tax reforms of 1992 and 1998
\citep{FagerengGuisoRing2025} provide exogenous shocks to the equity
premium that can be used for identification.  A diffusion model trained
on data spanning these reforms should detect changes in the drift
coefficient, and the \emph{form} of the change (uniform shift versus
state-dependent) tests the neutrality hypothesis of
\citet{Froeseth2026N}.

Additional natural experiments are available from other countries.
The 1989 Danish wealth tax reform \citep{JakobsenEtAl2020} provides
a clean before--after comparison for wealth accumulation dynamics.
Swiss cantonal variation \citep{BrulhartEtAl2022, MartiMartinezScheuer2023,
Burgherr2021} offers cross-sectional identification of how different
tax rates affect the stationary distribution.  Norwegian corporate
restructuring responses to valuation discounts \citep{BjornebyEtAl2023}
identify avoidance channels corresponding to the attenuated drift
of \eqref{eq:evasion}: the wedge between intended and realised drift
modification is directly observable.  Finally,
\citet{AlstadsaeterJohannesenZucman2019} show that tax evasion at the
top of the distribution is substantial (approximately 25\% of tax
liability for the top 0.01\%), implying that the observed $\pi(x,t)$
from tax registers underestimates true concentration and that the
evasion rate $\epsilon(x)$ of \Cref{sec:leakage} is empirically large.

\section{Discussion}\label{sec:discussion}

\subsection{What the framework provides}

The Fokker--Planck approach to redistribution offers four advantages
over existing frameworks.  First, it provides a \emph{unified language}
for comparing tax designs: proportional, progressive, transfers, and
caps all correspond to specific modifications of the same equation.
Second, it makes the \emph{speed--magnitude trade-off} explicit: the
spectral gap of the FP operator determines how fast the distribution
converges to its new steady state.  Third, it connects the
\emph{micro-level} dynamics (individual wealth accumulation) to the
\emph{macro-level} outcome (the wealth distribution) through a single
PDE.  Fourth, it models \emph{both redistribution channels}---market
(drift modification) and fiscal (source-sink terms)---within the same
equation, and the general equilibrium extension (\Cref{sec:ge})
captures their interaction through the production function: the tax
design affects aggregate capital, which feeds back into both market
returns and the revenue that funds the fiscal channel.

\subsection{The choice of target distribution}

The optimal drift design problem (\Cref{def:optimal}) requires
specifying a target distribution $\pi^*$.  The mathematical framework
is agnostic about this choice---it can accommodate any target---but
the choice itself is ultimately one of normative political philosophy,
not physics.

The framework does, however, make the consequences of different
normative positions mathematically precise.  A utilitarian planner
maximising total expected utility would choose $\pi^*$ to maximise
$\int u(e^x) \pi^*(x) \, \dd x$ for some concave utility function
$u$, subject to the production constraint.  A Rawlsian planner,
reasoning from behind the \emph{veil of ignorance}
\citep{Rawls1971}---uncertain of their own position in the wealth
distribution---would instead maximise the welfare of the worst-off
individual.  In the Fokker--Planck framework, the veil of ignorance
corresponds precisely to the initial condition: an agent who does not
yet know which realisation of the stochastic process they will
experience.  The Rawlsian maximin criterion then selects the drift
modification $\delta v(x)$ that maximises the minimum wealth level
in the steady-state distribution, subject to the GE constraint that
the production function must still be satisfied.

These are different objective functionals in the same optimal control
problem.  The utilitarian criterion typically yields a less aggressive
confining potential than the Rawlsian criterion, because it values
the surplus of the wealthy (through the concave $u$) whereas the
maximin criterion does not.  The framework makes this trade-off
quantitative: for each normative choice, it delivers the
corresponding optimal progressivity $\kappa^*$, convergence
timescale $t_{1/2}$, and GE feedback magnitude $\Delta r$.

This paper does not advocate for a particular normative position.
Its contribution is to provide the mathematical machinery that
translates any such position into a concrete policy design, and to
make the costs and timescales of that design explicit.

\subsection{Heterogeneous returns and the combined tax environment}

Two companion papers relax the homogeneous-returns assumption
(\Cref{remark:homogeneous_scope}) in ways that qualify---and in some
cases strengthen---the results above.

\citet{Froeseth2026H} introduces persistent heterogeneity in investor
ability, generating drift $v(z)$ that depends on an individual
parameter $z$.  Because a proportional tax preserves drift
\emph{differences}---the quantity $v(z_H) - v(z_L)$ is unchanged by a
uniform shift---it accelerates the reallocation of capital toward
higher-ability investors.  This ``use-it-or-lose-it'' mechanism
\citep{Guvenen2023} means the proportional tax \emph{is} redistributive
through the market channel under heterogeneous returns, even though the
redistribution paradox holds in the homogeneous limit.

The confining potential framework of \Cref{sec:confinement} remains
the appropriate tool for \emph{active} redistribution of the wealth
distribution.  But under heterogeneous returns, the baseline against
which the confining potential operates is itself redistributive: the
proportional tax component already generates reallocation, and the
progressive surcharge adds compression on top.  The optimal
progressivity $\kappa^*$ derived in \Cref{sec:optimal} should
therefore be interpreted as the additional confinement needed beyond
the passive reallocation effect.

\citet{Froeseth2026F} extends the neutrality framework to combined
flow-and-stock taxation, showing that corporate and dividend taxes
enter the Fokker--Planck equation as a rescaling of excess drift
differences by the factor $(1 - \tau_c)(1 - \tau_d)$, while the
wealth tax enters as a uniform shift.  The two effects are additively
separable under natural institutional conditions.  For the taxonomy
of \Cref{sec:taxonomy}, this means that in a combined tax system,
Class~2 modifications (progressive wealth tax) interact with the
flow-tax environment: the effective drift gap between asset classes
depends on both the progressivity parameter $\kappa$ and the
flow-tax rates $\tau_c, \tau_d$.  In particular, a revenue-neutral
shift from flow taxes to wealth taxes \emph{widens} ability-driven
drift gaps \citep{Froeseth2026H}, a consideration that enters the
optimal design problem when the tax mix is itself a policy variable.
That combined-instrument design problem is developed in companion
work: by \citet{Froeseth2026P} under homogeneous returns, and by
\citet{Froeseth2026HP} as the heterogeneous-returns synthesis.

\subsection{Limitations and extensions}

Several important extensions remain for future work.  The OU model
(linear confinement) is a simplification; realistic progressive tax
schedules have nonlinear rate structures that produce non-Gaussian
steady states.  \citet{Froeseth2026P} develops two such
extensions---smooth log-progressive and bracket schedules---within
the combined flow--stock framework; fully general schedule classes
remain open.  The leakage channels formalised in \Cref{sec:leakage}---migration
as a permeable boundary \eqref{eq:migration} and evasion as attenuated
drift \eqref{eq:evasion}---are treated here as exogenous.  A full
analysis would endogenise $\gamma$ and $\epsilon(x)$ as functions of
the tax design itself, coupling the leakage rates to the optimal
control problem.  The experience of the French ISF, where
\citet{Pichet2007} estimated substantial capital flight, and of the
Swiss cantons, where \citet{BrulhartEtAl2022} documented significant
behavioural responses, confirms that leakage is empirically important.

The empirical estimation of $v(x)$ and $D(x)$ from Norwegian data is
the natural next step, and will determine whether the GBM baseline is
an adequate starting point or whether state-dependent coefficients are
essential.  \citet{Zucman2024}'s proposal for a coordinated global
minimum wealth tax on billionaires would, by closing migration channels,
effectively set $\gamma \to 0$ in \eqref{eq:migration} and permit
stronger confinement than national policies alone.

\subsection{Connection to the wealth tax programme}

This paper is part of a series studying wealth taxation through the
Fokker--Planck framework.  \citet{Froeseth2026N} established that a
proportional market-value wealth tax is neutral with respect to
portfolio choice.  \citet{Froeseth2026E} identified the channels
through which neutrality breaks in practice.
\citet{Froeseth2026SP} reformulated both results in Fokker--Planck
language, revealing the drift-shift symmetry as the mathematical
content of neutrality.  \citet{Froeseth2026F} extended the
neutrality framework to combined flow-and-stock taxation, showing
additive separability of flow-tax and stock-tax distortions under
natural institutional conditions.  \citet{Froeseth2026H} introduced
heterogeneous investor ability, showing that the proportional tax is
simultaneously non-distortionary \emph{and} redistributive when
returns vary across investors---qualifying the redistribution paradox
established here to the homogeneous case.

The present paper develops the optimal drift design framework:
it characterises when and how the drift-shift symmetry must be
broken to achieve market-channel redistribution, formulates the
design problem as optimal control of the Fokker--Planck equation,
and models the general equilibrium interaction between the market
and fiscal channels.

Four further companions build directly on the design framework
developed here.  \citet{Froeseth2026D} characterises
minimum-distortion tax designs on the proportional class in
closed form, contrasting an information-theoretic with a
transport-geometric optimality criterion.  \citet{Froeseth2026P} lifts the progressive
drift-design problem to the combined flow--stock environment under
Gini and top-share criteria, including bracket schedules.
\citet{Froeseth2026HP} extends that synthesis to heterogeneous
returns.  \citet{Froeseth2026W} recasts the neutral class and
progressive optimality in 2-Wasserstein geometry---the
transport-geometric counterpart of the drift-design view taken
here.

A companion empirical paper will estimate the drift and diffusion
functions from Norwegian wealth register data using neural
stochastic differential equation methods, testing the theoretical
predictions developed here---including the neutrality hypothesis,
the general equilibrium feedback, and the leakage channels---against
observed wealth dynamics.

\subsection*{Acknowledgements}
The author acknowledges the use of Claude (Anthropic) for assistance with
literature review, \LaTeX{} typesetting, mathematical exposition, and
editorial refinement, and Lemma (Axiomatic AI) for review and proof
checking. All substantive arguments, economic reasoning, and conclusions
are the author's own.

\appendix

\section{Gini coefficient for standard distributions}
\label{app:gini}

For reference, we collect the Gini coefficient formulas for the
distributions appearing in the main text.

\textbf{Pareto distribution} with tail exponent $\alpha > 1$:
\begin{equation}
  \Gini_{\mathrm{Pareto}} = \frac{1}{2\alpha - 1} \,.
\end{equation}

\textbf{Log-normal distribution} with log-variance $\sigma_x^2$:
\begin{equation}
  \Gini_{\mathrm{LN}} = 2\Phi\!\left(\frac{\sigma_x}{\sqrt{2}}\right) - 1
  = \mathrm{erf}\!\left(\frac{\sigma_x}{2}\right) \,.
\end{equation}

\textbf{Exponential distribution} (Boltzmann--Gibbs, following
\citealp{DragulescuYakovenko2000}):
\begin{equation}
  \Gini_{\mathrm{exp}} = \frac{1}{2} \,.
\end{equation}

\textbf{Truncated Pareto distribution} with exponent $\alpha > 1$ and
upper bound $W_{\max}$, to leading order in
$u = (W_{\min}/W_{\max})^{\alpha - 1}$:
\begin{equation}
  \Gini_{\mathrm{trunc}} = \Gini_{\mathrm{Pareto}}
  \bigl[1 - u\bigr] + o(u) \,,
\end{equation}
so truncation reduces the Gini below its untruncated value by the
fraction $u$ at leading order.

\end{document}